\documentclass[graybox]{icmcta4}

\usepackage{mathptmx}       % selects Times Roman as basic font
\usepackage{helvet}         % selects Helvetica as sans-serif font
\usepackage{courier}        % selects Courier as typewriter font
\usepackage{type1cm}        % activate if the above 3 fonts are

\usepackage{amsmath}
\usepackage{amssymb}
\usepackage{makeidx}         % allows index generation
\usepackage{graphicx}        % standard LaTeX graphics tool
\usepackage{multicol}        % used for the two-column index
\usepackage[bottom]{footmisc}% places footnotes at page bottom

\newcommand{\field}[1]{\mathbb{#1}}

\newcommand{\F}{\field{F}}

\makeindex             % used for the subject index
\usepackage[latin1]{inputenc}       % used for the two-column index
\usepackage{epstopdf}       % used for the two-column index

\begin{document}

\title*{Fractional Repetition and Erasure Batch Codes}
\author{Natalia Silberstein}
\institute{N. Silberstein is with the Department of Computer Science, Technion, Haifa 32000,
%\at Technion -- Israel Institute of Technology, Haifa 32000,
Israel, \email{natalys@cs.technion.ac.il}\\
This research was supported in part by the Fine Fellowship and by the Israeli Science Foundation (ISF), Jerusalem, Israel, under Grant 10/12.}

%\and Name of Second Author \at Name, Address of Institute \email{name@email.address}
\maketitle

\abstract{
 Batch codes are a family of codes that represent a distributed storage system (DSS) of $n$ nodes so that any batch of $t$ data symbols can be retrieved by reading at most one symbol from each node. Fractional repetition codes are a family of codes for DSS that enable efficient uncoded repairs of failed nodes. In this work these two families of codes are combined to obtain fractional repetition batch (FRB) codes which provide both uncoded repairs and parallel reads of subsets of stored symbols. In addition, new batch codes which can tolerate node failures are considered. This new family of batch codes is called erasure combinatorial batch codes (ECBCs). Some properties of FRB codes and ECBCs and examples of their constructions based on transversal designs and affine planes are presented.
}
%\begin{comment}Each chapter should be preceded by an abstract (10--15 lines long) that summarizes the content. The abstract will appear \textit{online} at \url{www.SpringerLink.com} and be available with unrestricted access. This allows unregistered users to read the abstract as a teaser for the complete chapter. As a general rule the abstracts will not appear in the printed version of your book unless it is the style of your particular book or that of the series to which your book belongs.\newline\indent
%Please use the 'starred' version of the new Springer \texttt{abstract} command for typesetting the text of the online abstracts (cf. source file of this chapter template \texttt{abstract}) and include them with the source files of your manuscript. Use the plain \texttt{abstract} command if the abstract is also to appear in the printed version of the book.
%\end{comment}

%=====================================================================================================================================
\vspace{-.2cm}
\section{Introduction}
\label{sec:Intro}
In distributed storage systems (DSS) information is stored across a network of nodes in such a way that a user (data collector) can retrieve the stored data even if some system nodes fail.
To provide reliability against node failures, data redundancy based on different types of erasure codes is introduced in such systems.
Moreover, to provide an efficient repair of a single failed node (the most common case in DSS), a new family of erasure codes for DSS, called \emph{regenerating} codes, was presented in~\cite{dimakis}. Two types of regenerating codes, \emph{minimum storage regenerating} (MSR) and \emph{minimum bandwidth regenerating} (MBR)~\cite{dimakis} codes, were introduced to optimize the storage overhead and repair bandwidth, respectively (for constructions see~\cite{dimakis,DRWS11,RSK11,SRKR12_transfer} and references therein).
In particular, a regenerating code $C$ is used to store a file on $n$ nodes, where each node stores $\alpha$ symbols from a finite field $\F_q$, such that a data collector can recover the stored file from any set of $k<n$ nodes. A single failed node can be repaired by downloading $\beta\leq \alpha$ symbols from any node in a set of size $d$, $k\leq d\leq n-1$, of surviving nodes. Note that any random set of $d$ nodes can be used to repair a failed node.

\emph{Fractional repetition} (FR) codes~\cite{ElrRam} are a family of codes for DSS which allow for uncoded repairs (no decoding is needed), while relaxing
the requirement of random $d$-set for repairs by making it table based instead. This relaxation allows for increasing the amount of data that can be stored by using FR codes when compared to MBR codes, while having the same repair bandwidth.
When an $(n,k,\alpha,\rho)$ FR code $C$ is used to store a
file $\bf f$ $ \in \mathbb{F}_q^M$ of size $M$, $\bf f$ is first encoded  to a codeword $c_{\bf f}$ of a $(\theta,M)$ maximum distance separable (MDS) code~\cite{MWSl78}, with $\theta=n \alpha/\rho$. Next, $\theta$ symbols of the MDS codeword $c_{\bf f}$ are placed on $n$ nodes, each of size $\alpha$, as follows. Let $N_1,\ldots, N_n$ be a collection of subsets of size $\alpha$ of the set $[\theta]:=\{1,2,\ldots,\theta\}$, such that every element in $[\theta]$ appears in exactly $\rho$ subsets. Then node $i$ stores the symbols of $c_{\bf f}$ indexed by the subset $N_i$. An FR code should satisfy the requirement that from any set of $k$ nodes it is possible to reconstruct the stored file, that is, $M=\min_{|I|=k}|\cup_{i\in I}N_i|$.  Note that for FR codes it holds that $\alpha=d$ and $\beta=1$, since when some node $i$ fails, it can be repaired by using $\alpha$ other nodes which store common symbols with node $i$.
Constructions of FR codes based on different types of regular graphs and combinatorial designs can be found in~\cite{ElrRam,resolvFR,dress,FR_TD}.

\emph{Batch codes}~\cite{batchAp} are a family of codes for DSS which store $\theta$ (encoded) data symbols on $n$ system nodes in such a way that any batch of $t$ data symbols can be decoded by reading at most one symbol from each node, while keeping the total storage over all $n$ nodes equal to $N$.
A $\rho$-\emph{uniform} \emph{combinatorial batch code} (CBC), denoted by $\rho-(\theta,N=\rho \theta,t,n)$,  is a batch code where
each node stores a subset of data symbols, that is decoding is performed only by reading items from the nodes, and each symbol is stored in exactly $\rho$ nodes~\cite{batchAp,Paterson}.  These codes were studied in~\cite{batchBounds,batchTuBu,batchAp,Paterson,batchTD}.
%in~\cite{PSW08,BRR11,BaBh12,BuTu13,SiGa13}.

In this work, we consider two new families of codes for DSS.
The first family, called \emph{fractional repetition batch} (FRB) codes, is based on the combination of FR and combinatorial batch codes and hence has the properties of both FR and batch codes simultaneously: FRB codes allow for uncoded efficient repairs and load balancing in partial data reconstruction which can be performed by several users independently and in parallel.
%We provide several examples of such codes and analyze their properties.
The second family of codes, called \emph{erasure combinatorial batch codes} (ECBCs), allow for recovery of any batch of $t$ data symbols even in presence of nodes failures, by reading at most one symbol from the remaining available nodes. ECBCs generalize the original batch codes~\cite{batchAp,Paterson} which require \textit{all} the nodes in a system to be always available for accessing their stored data.
We analyze the properties of incidence matrices of FRB codes and ECBCs and present the necessary and sufficient conditions on the structure of these codes. We provide constructions for FRB codes and ECBCs based on transversal designs and affine planes.

%we consider a new family of codes for DSS, which we call \emph{fractional repetition batch} (FRB) codes,
%that have both the properties of FR and batch codes simultaneously. FRB codes allow for uncoded efficient repairs and load balancing in partial data reconstruction which can be performed by several users independently and in parallel. We provide several examples of such codes and analyze their properties.
%In addition, motivated by the application of erasure codes in DSS, we propose new batch codes which can tolerate node erasures and present bounds and constructions for this type of batch codes.

The rest of this paper is organized as follows. In Section~\ref{sec:FR Batch}
we define FRB codes, consider  properties of their incidence matrices and provide some examples of their constructions. In Section~\ref{sec:erasure batch} we define ECBCs, discuss their properties and describe codes based on affine planes and transversal designs. Conclusions and problems
for future research are given in Section~\ref{sec:conclusion}.
%=====================================================================================================================================
\vspace{-.2cm}
\section{Fractional Repetition Batch Codes}
\label{sec:FR Batch}

In this section we consider a new family of codes for DSS, called FRB codes, which combine the properties of both FR and combinatorial batch codes.

Let $\bf f$ $\in \F_q^M$ be a file of size $M$ and let $c_{\bf f}\in \F_q^{\theta}$ be a codeword of an $(\theta,M)$ MDS code which encodes the data $\bf f$.
Let $\{N_1,\ldots,N_n\}$ be a collection of $\alpha$-subsets of a set $[\theta]$.
 A $\rho-(n,M,k,\alpha,t)$ \emph{fractional repetition batch} (FRB) code $C$ represents a system of $n$ nodes with the following properties:
\begin{enumerate}
  \item Every node $i$, $1\leq i\leq n$, stores $\alpha$ symbols of $c_{\bf f}$ indexed by $N_i$;
  \item Every symbol of $c_{\bf f}$ is stored on $\rho$ nodes;
  \item From any set of $k$ nodes it is possible to reconstruct the stored file $\bf f$, in other words, $M=\min_{|I|=k}|\cup_{i\in I}N_i|$;
  \item Any batch of $t$ symbols from $c_{\bf f}$ can be retrieved by downloading at most one symbol from each node.
\end{enumerate}

Note that the total storage over all $n$ nodes needed to store a file $\bf f$ equals to $n\alpha=\theta\rho$. The general coding scheme for an FRB code is shown in Fig.~\ref{fig:FRscheme}.

\begin{figure*}[h]
\centering
%\begin{center}
\includegraphics[width=0.85\textwidth]{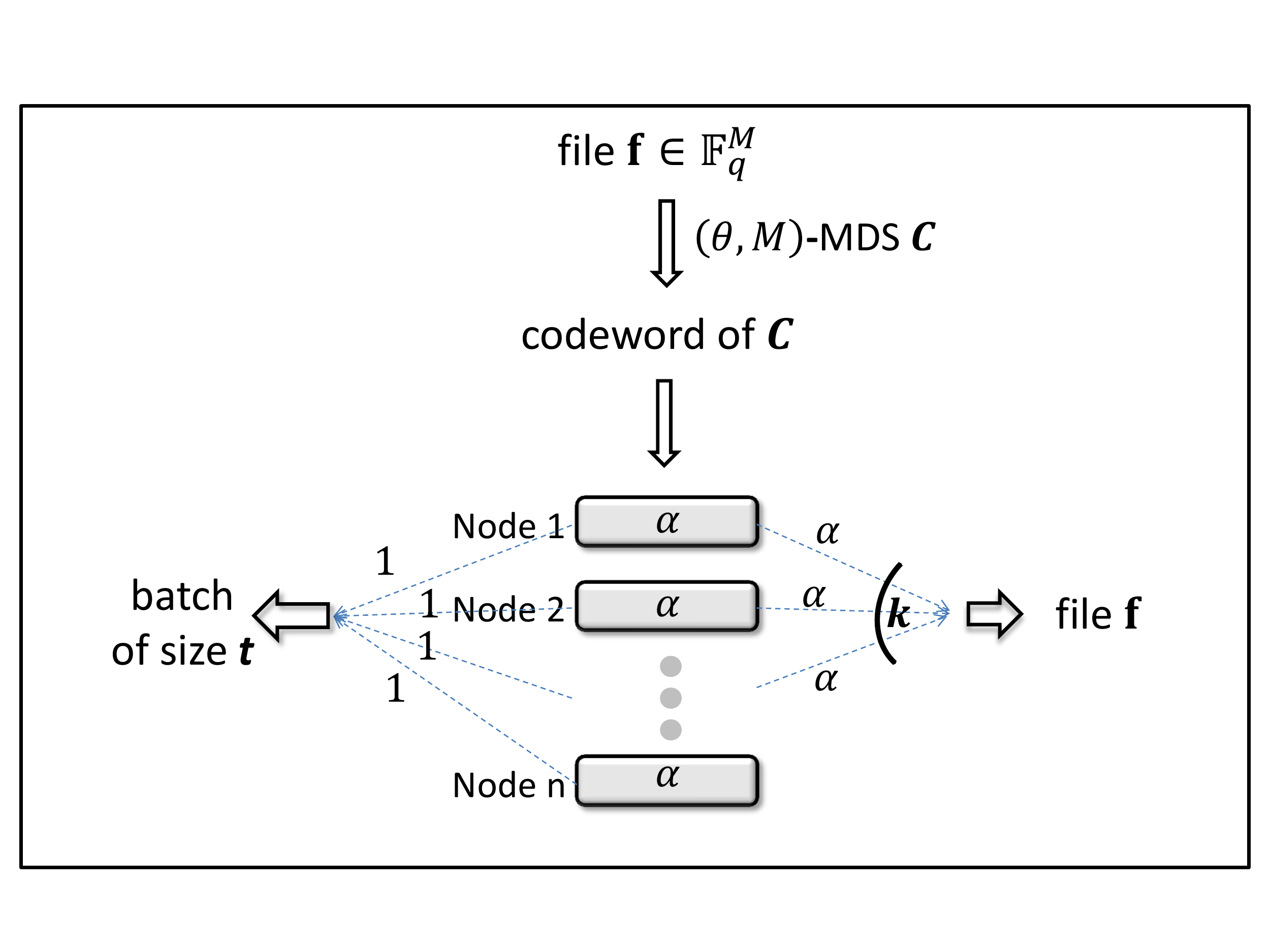}
\caption{The coding scheme for an FRB code}\label{fig:FRscheme}
%\end{center}
\end{figure*}

\begin{remark}
Note that while in a classical batch code any $t$ \emph{data} symbols can be retrieved, in a FRB code any batch of $t$ \emph{coded} symbols can be retrieved.  In particular, when a systematic MDS code is chosen for an FRB code, the data symbols can be easily retrieved.
\end{remark}

Now we consider the matrix representation of FRB codes which follows from the matrix representation of FR and combinatorial batch codes.
The incidence matrix of a $\rho-(n,M,k,\alpha,t)$ FRB code $C$, denoted by $\textbf{I}(C)$, is  a binary $n\times \theta$ matrix with rows and columns indexed by the nodes and symbols of an MDS codeword, respectively, such that
$(\textbf{I}(C))_{i,j}=1$ if and only if node $i$ contains symbol $j$ of $c_{\bf f}$. In other words, the
$i$th row of $\textbf{I}(C)$ is the incidence vector of the set $N_i$. Note that the number of ones in each row is $\alpha$ and the number of ones in each column is $\rho$ in this matrix.

In the following, we obtain the necessary and sufficient conditions on a binary matrix to be the incidence matrix of an FRB code.
Let $A$ be a binary matrix, and let $S$ and $T$ be some subsets of rows and columns of $A$, respectively. Let $A_{S,T}$ be a submatrix of $A$ with rows and columns indexed by $S$ and $T$. We say that a set $T$ of columns \emph{covers} a set $S$ of rows if there is no all-zero row in $A_{S,T}$. Similarly, a set $S$ of rows \emph{covers} a set $T$ of columns if there is no all-zero column in $A_{S,T}$.

The next theorem follows from the properties of incidence matrices for combinatorial batch and FR codes (see~\cite{batchBounds,Paterson,FR_TD} for details).
\begin{theorem}
\label{thm:incidence matrix}
An $n\times \theta$ binary matrix $A$ with $\alpha$ \emph{ones} in each row and $\rho$ \emph{ones} in each column is the incidence matrix of a $\rho-(n,M,k,\alpha,t)$ FRB code  if and only if the following two conditions hold:
\begin{enumerate}
  \item Any $i$ columns of $A$, $1\leq i\leq t$, cover at least $i$ rows;
  \item Any $k$ rows of $A$ cover at least $M$ columns.
\end{enumerate}
\end{theorem}
If we consider the incidence matrix of an FRB code as the biadjacency matrix of a bipartite graph, where the left vertex set $L$
corresponds to the rows (nodes) and the right vertex set $R$ corresponds to the columns (codeword symbols) of the matrix, then the conditions of Theorem~\ref{thm:incidence matrix} can be formulated as follows.
\begin{corollary}
\label{cor:bipartite graph}
A biadjacency matrix of a bipartite graph $G=(L\cup R, E)$, $|L|=n$, $|R|=\theta$, with the left degree $\alpha$ and right degree $\rho$, is the incidence matrix of a $\rho-(n,M,k,\alpha,t)$ FRB code  if and only if the following two conditions hold:
%requirements are satisfied:
\begin{enumerate}
  \item Any subset $T\subseteq R$ of at most $t$ vertices has at least $|T|$ neighbours in $L$;
  \item Any subset $S\subseteq L$ of $k$ vertices has at least $M$ neighbours in $R$.
\end{enumerate}
\end{corollary}
\begin{remark} The construction of batch codes based on unbalanced expander graphs was proposed in~\cite{batchAp}. To construct an FRB code, we need a bipartite expander with two different expansion factors, 1 and $M/k$, for two sides $R$ and $L$ of a graph, respectively.
\end{remark}
%=====================================================================================================================================
\vspace{-0.5cm}
\subsection{Constructions of FRB codes}
In this subsection, we consider constructions of FRB codes based on optimal FR codes and optimal uniform CBCs.
We say that an FR code is an optimal code if it can store a file of maximum size, i.e. it maximizes $M=M(n,k,\alpha,\rho)$ (see~\cite{ElrRam,FR_TD}  for details).
%For example, FR codes based on Tur\'an graphs, graphs with large girth and on transversal designs (see []), are optimal FR codes.
We say that a uniform  combinatorial batch code is an optimal code if it stores the maximum number of symbols, i.e., it maximizes $\theta=\theta(n,\rho,t)$ (see~\cite{batchBounds,Paterson,batchTD}).

It was proved recently~\cite{batchTD} that combinatorial batch codes based on some transversal designs are (near) optimal CBCs. Moreover, it was shown that FR codes based on transversal designs are optimal FR codes~\cite{FR_TD}. Therefore, it is natural to consider FRB codes based on transversal designs.
%First, we present a definition of a transversal design.

A \emph{transversal design} (TD) of group size $h$ and block size $\ell$,  denoted by $\text{TD}(\ell, h)$,
%$\text{TD}_{\lambda}(\ell, h)$
is a triple $(\mathcal{P},\mathcal{G},\mathcal{B})$, where
\begin{enumerate}
\item $\mathcal{P}$ is a set of $\ell h$ \emph{points};
\item $\mathcal{G}$ is a partition of $\mathcal{P}$ into $\ell$ sets
(\emph{groups}), each one of size $h$;
\item $\mathcal{B}$ is a collection of $\ell$-subsets of $\mathcal{P}$
(\emph{blocks});
\item each block meets each group in exactly one point;
\item any pair of points from different groups is contained in exactly one block.
\end{enumerate}
It follows from the definition of TD that the number of blocks in $\text{TD}(\ell, h)$ is $h^2$ and the number of blocks that contain a given point is $h$~\cite{Anderson}. The incidence matrix $\textbf{I}_{TD}$ of $\text{TD}(\ell, h)$ is  the $\ell h\times h^2$ binary matrix where columns are incidence vectors of the blocks.
A $\text{TD}(\ell,h)$ is called \emph{resolvable} if the set $\mathcal{B}$
can be partitioned into sets $\mathcal{B}_1,...,\mathcal{B}_h$, each one contains $h$ blocks,
such that each element of $\mathcal{P}$ is contained in exactly one block of
each $\mathcal{B}_i$.
%The sets $\mathcal{B}_1,...,\mathcal{B}_h$ are called \emph{parallel classes}.
Resolvable $\text{TD}(\ell,h)$ is known to exist for any $\ell \leq h$ and prime power $h$~\cite{Anderson}.
%The incidence matrix of a resolvable TD can be written in a blocks form, where every block is a $h\times h$ permutation matrix which corresponds to $h$ points of a group and $h$ blocks of a parallel class.

Next we consider an FRB code $C_{\textmd{TD}}$ such that its incidence matrix is the incidence matrix of TD. Based on the properties of uniform CBCs and FR codes constructed from different TDs~\cite{FR_TD,batchTD}, we obtain the following result.
\begin{theorem}
\label{thm:FRB from TD}
$~$
\begin{enumerate}
  \item Let $\textmd{TD}(2,\alpha)$ be a TD with $\alpha>2$. Then $C_{\textmd{TD}}$ is a $2-(2\alpha, M,k,\alpha,5)$ FRB code with
  $M=k\alpha-\left\lfloor\frac{k^2}{4}\right\rfloor$.
  \item Let $\textmd{TD}(\alpha-1,\alpha)$ be a resolvable TD,  for a prime power $\alpha$. Then $C_{\textmd{TD}}$ is a ${(\alpha-1)-(\alpha^2-\alpha, M,k,\alpha,\alpha^2-\alpha-1)}$ FRB code with
  $M\geq k\alpha -\binom{k}{2}+(\alpha-1)\binom{x}{2}+xy$, where $x,y\geq 0$ are integers which satisfy $k=x(\alpha-1)+y$, $y\leq \alpha-2$.
\end{enumerate}
\end{theorem}
\begin{example} We consider the FRB code based on $\textmd{TD}(3,4)$. By Theorem~\ref{thm:FRB from TD}, for $k=4$ we have a $3-(12,11,4,4,11)$ FRB code, which stores a file of size $11$ and which allows for parallel reads of any (coded) $11$ symbols.
%For $k=3$, we have a $3-(12,9,3,4,11)$ FRB code, which stores a file of size $9$ and which allows of parallel reads of any (coded) $11$ symbols.
\end{example}
In general, when a given FR code is considered as a batch code, determining its parameter $t$ (the number of symbols that can be read in parallel) is a nontrivial task. Similarly,  for a given batch code it is difficult to find the parameter $M$  (the file size) for any $k$. In the following, we consider a FRB code based on $\textmd{TD}(3,\alpha)$, where every symbol is replicated 3 times. For this code, the parameter $M$ is given in~\cite{FR_TD}. We obtain the upper and lower bounds on $t$ in the following theorem.
\begin{theorem}
\label{thm:TD(3)}
The FRB code based on $\textmd{TD}(3,\alpha)$ is a $3-(3\alpha, M, k, \alpha, t)$ code, where $6\leq t\leq 2\alpha+1$ for $\alpha\geq 7$ and $t=12$ for $\alpha=5$. The file size is given by $M=k\alpha-\binom{k}{2}+3\binom{x}{2}+xy$, for $x,y\geq 0$ such that $k=3x+y$ and $y\leq 2$.
\end{theorem}
\begin{proof}
The parameters $n, \rho$ and $M$ follow from the properties of $\textmd{TD}(3,\alpha)$ and FR codes based on TDs~\cite{FR_TD}. The lower bound on $t$ follows from Theorem~\ref{thm:FRB from TD}.1. To prove the upper bound on $t$ one can consider a specific structure of an incidence matrix for TD and show that there are $2\alpha+2$ columns that cover only $2\alpha+1$ rows.
\qed
\end{proof}
%\begin{example} An FRB code based on $\textmd{TD}(3,5)$ is a $3-(15,4,15,5,12)$ code.
%\end{example}

In the rest of this section we consider FRB codes obtained from affine planes. The optimality of uniform combinatorial batch codes based on affine planes was proved in~\cite{batchTD}.

 An \emph{affine plane} of order $s$, denoted by $A(s)$, is  a set system $(X,\mathcal{B})$, where $X$ is a set of $|X|=s^2$ \emph{points}, $\mathcal{B}$ is a a collection of $s$-subsets (\emph{blocks}) of $X$ of size $|\mathcal{B}|=s(s+1)$, such that each pair of points in $X$ occur together in exactly one block of~$\mathcal{B}$.
 An affine plane is called \emph{resolvable}, if the \emph{set} $\mathcal{B}$ can be partitioned into $s+1$ sets of size $s$, called parallel classes, such that every element of $X$ is contained in exactly one block of each class. It is well known~\cite{Anderson} that if $q$ is a prime power, then there exists a resolvable affine plane $A(q)$.
 %It is well known~\cite{Anderson} that if $q$ is a prime power, then there exists an affine plane of order $q$.
\begin{theorem}
\label{thm:affine}
Let  $A(q)$ be an affine plane and let $\textbf{I}(A)$ be its $q^2\times (q^2+q)$ incidence matrix. Then the FRB code $C_A$ with the incidence matrix equal to $\textbf{I}(A)$ is a $q-(q^2, k(q+1)-\binom{k}{2},k,q+1,q^2)$ FRB code.
\end{theorem}
\begin{proof}
The parameters $\rho, n,\alpha$ and $t$ follow from the properties of the batch code based on $A(q)$~\cite{batchTD}.
Since any two points of $A(q)$ belong to exactly one block and hence any two rows of $\textbf{I}(A)$ intersect, it follows that the file size is $k(q+1)-\binom{k}{2}$.
\qed
\end{proof}
%?????FRB codes based on graphs with large girth?
%=====================================================================================================================================
\vspace{-0.4cm}
\section{Erasure Combinatorial Batch Codes}
\label{sec:erasure batch}
In this section we consider uniform combinatorial batch codes which can tolerate node failures (erasures). We call such batch codes \emph{erasure batch} codes.
Specifically, we define a $\rho-(\theta,N=\rho\theta,t,n,\Delta)$ \emph{uniform erasure combinatorial batch} code (ECBC)
to be a code which stores $\theta$ data symbols on $n$ nodes, such that each symbol is stored on $\rho$ nodes and for any given set of
$\Delta$ failed nodes, any batch of $t$ symbols can be retrieved by reading at most one symbol from each one of $n-\Delta$ available nodes, while keeping the total storage equal to $N$. Note that it should hold that $\Delta\leq \rho-1$.

\begin{remark}
Note that if any set of $\Delta$ nodes contains at most $t$  different symbols, then it is possible to correct any $\Delta$ erasures, i.e., to repair $\Delta$ failed nodes by reading at most one symbol from  every available node.
\end{remark}

Similarly to Theorem~\ref{thm:incidence matrix}, we provide the necessary and sufficient conditions on a binary matrix to be the incidence matrix of a uniform ECBC.

\begin{theorem}
\label{thm:erasure batch}
An $n\times \theta$ binary matrix $A$ with $\rho$ \emph{ones} in each column is the incidence matrix of a $\rho-(\theta, N, t,n,\Delta)$ uniform ECBC  if and only if
any $i$ columns of $A$, $1\leq i\leq t$, cover at least $i+\Delta$ rows.
\end{theorem}

Based on Theorem~\ref{thm:erasure batch} and resolvability of %an affine plane
$A(q)$~\cite{Anderson} we have the following result.
\begin{theorem}
\label{erasure batch affine}
Let  $A(q)$ be an affine plane and let $\textbf{I}(A)$ be its $q^2\times (q^2+q)$ incidence matrix. Then the code $C^{E}_A$ with the incidence matrix equal to $\textbf{I}(A)$ is a $q-(q^2+q,q^3+q^2,t,q^2,q-1)$ uniform ECBC, where $\frac{q^2-q+2}{2}\leq t\leq q^2-q$.
\end{theorem}

\begin{proof} The parameters $\rho, \theta, N, n$ follow from the properties of $A(q)$, and $\Delta$ is the largest possible. To prove the upper bound on $t$ we consider a set of erased nodes which correspond to $q-1$ points of  a block $b$ of $A(q)$. Let $p\in b$ be the point which was not erased. If we take one block in the parallel class which contains $b$ and $q-1$ blocks which do not contain $p$ in each one of $q$ other parallel classes, then the corresponding $q^2-q+1$ columns of $\textbf{I}(A)$ cover at most $q^2-1$ rows, thus by Theorem~\ref{thm:erasure batch}, $t\leq q^2-q$. To prove the lower bound on $t$ we note that any $q$ columns of $\textbf{I}(A)$ cover at least $q^2-\binom{q}{2}$ rows (since there are $q$ blocks of $A(q)$ which pairwise intersect).
Then since $\frac{q^2-q+2}{2}\geq q$ for $q\geq 2$, any $i$ columns, where  $q\leq i\leq \frac{q^2-q+2}{2}$, cover at least $q^2-\binom{q}{2}=\frac{q^2-q+2}{2}+(q-1)\geq i+(q-1)$ rows. For $i\leq q-1$ it holds that any $i$ columns cover at least $iq-\binom{i}{2}\geq i+(q-1)$ rows, which completes the proof.
\qed
\end{proof}

Now we consider a uniform ECBC $C_{\textmd{TD}}^E$ based on a transversal design, i.e., the code with the incidence matrix equal to the incidence matrix of TD. Similarly to Theorems~\ref{thm:FRB from TD} and~\ref{thm:TD(3)} one can prove the following result.
\begin{theorem}
\label{erasure batch TD}
{~}
\begin{itemize}
  \item Let $\text{TD}(2,\alpha)$ be a TD with $\alpha>2$. Then the code $C_{\textmd{TD}}^E$ is a $2-(\alpha^2, 2\alpha^2,3,2\alpha,1)$ uniform ECBC.
  \item Let $\text{TD}(3,\alpha)$ be a TD with $\alpha>3$. Then the code $C_{\textmd{TD}}^E$ is a $3-(\alpha^2, 3\alpha^2,t,3\alpha,2)$ uniform ECBC, where $4\leq t\leq 2\alpha-2$ for $\alpha\geq 6$, $t=9$ for $\alpha=5$, and $t=8$ for $\alpha=4$.
\end{itemize}
\end{theorem}
%=====================================================================================================================================
\vspace{-0.5cm}
\section{Conclusion and Future Work}
%\section{Conclusion}
\label{sec:conclusion}
This paper introduces two new families of erasure codes for distributed storage systems, namely fractional repetition batch codes and uniform erasure combinatorial batch codes. FRB codes have the properties of both FR and batch codes allowing for uncoded repairs of failed system nodes and parallel reads of subsets of data symbols. Uniform ECBCs have the properties of combinatorial batch codes even in presence of system nodes failures.
We provide the matrix description of these codes and present constructions based on transversal designs and affine planes.

We conclude with a list of open problems for future research.
\begin{enumerate}
  \item Find an upper bound on $t$ and $M$ given  other parameters $\{n,\rho,\alpha,k\}$ for an FRB code;
  \item Given  the set of parameters $\{n,\rho,\alpha,k\}$, construct a $\rho-(n,M,k,\alpha,t)$ FRB code with the maximum $M$ and $t$;
  \item Find the exact values of $t$ for FRB codes and ECBCs based on transversal designs and affine planes.
 \end{enumerate}
%\vspace{-0.4cm}
\begin{acknowledgement}
The author thanks Tuvi Etzion and Mark Silberstein for the valuable discussions.
The author also wishes to thank COST Action IC1104 "Random Network Coding and Designs over GF(q)" on travel support to present this work.
\end{acknowledgement}
% \vspace{-0.4cm}
\nocite{*}
% Use \nocite to cite every item in the `bibtex' environment
% or in the bibliography database file (with `bib' extension)
% You may add references in three different ways.
% The first method shall be required for final submission.
%
% 1. Write the bibtex entries inside the `bibtex' environment;
%    this will produce a `bib' file which can be processed by BibTeX
%    (most TeX/LaTeX editors compile it automatically).
%
% 2. Use the `standard' \bibliography commands to produce references
%    based on your BibTeX database files (with `bib' extension).
%
% 3. Include references using the `thebibliography' environment.
%
% Examples follow.
% 1st method: `bibtex' environment
% Observe that `eprint' automatically adds the string `arXiv', so remove
% it from the output of arXivToBibTeX - http://www.crcg.de/arXivToBibTeX
\begin{bibtex}
@unpublished{batchTD,
author       = {N. Silberstein and A. G\'al},
title        = {Optimal combinatorial batch codes based on block designs},
year         = {2013},
eprint       = {1312.5505},
}
@unpublished{FR_TD,
author       = {N. Silberstein and T. Etzion},
title        = {Optimal fractional repetition codes},
year         = {2014},
eprint       = {1401.4734},
}
@article{Paterson,
year={2009},
journal={Advances in Mathematics of Communications},
volume={3},
number={1},
title={Combinatorial batch codes},
publisher={Springer-Verlag},
author={Paterson, Maura B. and Stinson, Douglas R. and Wei, Ruizhong},
pages={13-27},
%doi = {10.1007/s00222-005-0473-9},
}

@ARTICLE{dimakis,
  author = {Dimakis, A.G. and Godfrey, P.B. and Wu, Y. and Wainwright, M.J. and
	Ramchandran, K.},
  title = {Network Coding for Distributed Storage Systems},
  journal = {Information Theory, IEEE Trans. on},
  year = {2010},
  volume = {56},
  pages = {4539-4551},
  number = {9},
  %doi = {10.1109/TIT.2010.2054295},
  issn = {0018-9448},
  keywords = {error correction codes;network coding;storage area networks;data centers;distributed
	storage systems;encoded fragments;erasure coded system;network coding;peer-to-peer
	storage systems;regenerating codes;repair bandwidth;wireless networks;Bandwidth;Communication
	system control;Computer science;Distributed databases;Encoding;Helium;Maintenance
	engineering;Network coding;Peer to peer computing;Redundancy;Statistics;Telecommunication
	network reliability;Wireless networks;Distributed storage;network
	coding;peer-to-peer storage;regenerating codes}
}

@ARTICLE{RSK11,
  author = {Rashmi, K. V. and Shah, N.B. and Kumar, P.V.},
  title = {Optimal Exact-Regenerating Codes for Distributed Storage at the MSR
	and MBR Points via a Product-Matrix Construction},
  journal = {Information Theory, IEEE Trans. on},
  year = {2011},
  volume = {57},
  pages = {5227-5239},
  number = {8},
  %doi = {10.1109/TIT.2011.2159049},
  issn = {0018-9448},
  keywords = {codes;matrix algebra;storage management;distributed storage codes;erasure
	codes;minimum bandwidth regenerating codes;minimum storage regenerating
	codes;optimal exact-regenerating codes;product-matrix construction;product-matrix
	framework;Bandwidth;Context;Encoding;Joining processes;Maintenance
	engineering;Symmetric matrices;Systematics;Distributed storage;interference
	alignment;network coding;node repair;partial data recovery;product-matrix
	framework;regenerating codes}
}

@inproceedings{DRWS11,
  author = {A.~G.~Dimakis and K.~Ramchandran and Y.~Wu and C.~Suh},
  title = {A survey on network codes for distributed storage},
  booktitle = {Proc. of the IEEE},
  year = {2011},
  pages = {476--489},
  owner = {AnkitSingh},
  timestamp = {2013.10.10}
}

@inproceedings{ElrRam,
  author = {El Rouayheb, S. and Ramchandran, K.},
  title = {{F}ractional repetition codes for repair in distributed storage systems},
  booktitle = {Proc. 48th Annual Allerton Conf. on Communication, Control, and Computing
	(Allerton)},
  year = {2010},
  pages = {1510 -1517},
 % address = {Urbana-Champaign, IL},
  month = {Sep.}
}

@inproceedings{dress,
  author    = {Sameer Pawar and
               Nima Noorshams and
               El Rouayheb, S. and
               Kannan Ramchandran},
  title     = {DRESS codes for the storage cloud: Simple randomized constructions},
  booktitle = {Proc. 2011 IEEE Int. Symp. on Information Theory, ISIT 2011},
  year      = {2011},
  pages     = {2338-2342},
  %ee        = {http://dx.doi.org/10.1109/ISIT.2011.6033980},
  %crossref  = {DBLP:conf/isit/2011},
  %bibsource = {DBLP, http://dblp.uni-trier.de}
}

@ARTICLE{SRKR12_transfer,
  author = {Shah, N.B. and Rashmi, K. V. and Kumar, P.V. and Ramchandran, K.},
  title = {Distributed Storage Codes With Repair-by-Transfer and Nonachievability
	of Interior Points on the Storage-Bandwidth Tradeoff},
  journal = {Information Theory, IEEE Trans. on},
  year = {2012},
  volume = {58},
  pages = {1837-1852},
  number = {3},
  %doi = {10.1109/TIT.2011.2173792},
  issn = {0018-9448},
  keywords = {Reed-Solomon codes;source coding;Reed-Solomon codes;arithmetic operations;distributed
	storage codes;exact-repair code;helper node pooling;k nodes;repair-by-transfer;storage-bandwidth
	tradeoff;Bandwidth;Complexity theory;Distributed databases;Encoding;Joining
	processes;Maintenance engineering;Systematics;Distributed storage;minimum
	bandwidth;node repair;regenerating codes;storage versus repair-bandwidth
	tradeoff}
}

@article{batchBounds,
year={2012},
journal={Advances in Mathematics of Communications},
volume={3},
number={1},
title={Combinatorial Batch Codes: A Lower Bound and Optimal Constructions},
publisher={Springer-Verlag},
author={Bhattacharya, Srimanta and Ruj, Sushmita and Roy, Bimal K.},
pages={165-174}
%doi = {10.1007/s00222-005-0473-9},
}

@inproceedings{batchAp,
  author = {Ishai, Yuval and Kushilevitz, Eyal and Ostrovsky, Rafail and Sahai, Amit},
  title = {Batch codes and their applications},
  booktitle = {Proc. 36th annual ACM symp. on Theory of computing STOC '04},
  year = {2004},
  pages = {262-271},
  }

%@inproceedings{batchAp,
%  author    = {Ishai, Yuval and Kushilevitz, Eyal and Ostrovsky, Rafail and Sahai, Amit},
%  title     = {Batch codes and their applications},
%  booktitle = {Proc. 36th annual ACM symp. on Theory of computing STOC '04},
%  year      = {2004},
%  pages     = {262-271},
%  %ee        = {http://dx.doi.org/10.1109/ISIT.2011.6033980},
%  %crossref  = {DBLP:conf/isit/2011},
%  %bibsource = {DBLP, http://dblp.uni-trier.de}
%}

@article{batchTuBu,
  author    = {Csilla Bujt{\'a}s and
               Zsolt Tuza},
  title     = {Optimal batch codes: Many items or low retrieval requirement},
  journal   = {Advances in Mathematics of Communications},
  volume    = {5},
  number    = {3},
  year      = {2011},
  pages     = {529-541},
  ee        = {http://dx.doi.org/10.3934/amc.2011.5.529},
  bibsource = {DBLP, http://dblp.uni-trier.de}
}

@inproceedings{resolvFR,
  author    = {Oktay Olmez and
               Aditya Ramamoorthy},
  title     = {Repairable replication-based storage systems using resolvable
               designs},
  booktitle = {Proc. 50th Annual Allerton Conf. on Communication, Control, and Computing
	(Allerton)},
  year      = {2012},
  pages     = {1174-1181},
  ee        = {http://dx.doi.org/10.1109/Allerton.2012.6483351},
  crossref  = {DBLP:conf/allerton/2012},
  bibsource = {DBLP, http://dblp.uni-trier.de}
}
@BOOK{MWSl78,
  title = {The theory of error-correcting codes},
  publisher = {North-Holland},
  year = {1978},
  author = {F. J. MacWilliams and N. J. A. Sloane},
  %owner = {Ankit Singh Rawat},
  %timestamp = {2013.05.13}
}

@book{Anderson,
author="I. Anderson",
title="{Combinatorial designs and tournaments.}",
publisher="Clarendon Press",
address="Oxford",
year="1997",
%series="Graduate Studies in Mathematics",
%volume=145,
}
\end{bibtex}
%
%% 2nd method: \bibliography + \bibliographystyle
%
%%\bibliography{your bibliography database file}
%%\bibliographystyle{icmcta4}
%
%% 3rd method: (discouraged) `thebibliography' environment
%
%%\begin{thebibliography}{1}
%%
%%\bibitem{CarlssonGoldfarb}
%%Carlsson, G., Goldfarb, B.: Algebraic $k$-theory of geometric groups (2013).
%%\newblock \href {http://arxiv.org/abs/1305.3349} {\path{arXiv:1305.3349}}
%%
%%\bibitem{Cortinas}
%%Cortiñas, G.: The obstruction to excision in $k$-theory and in cyclic homology.
%%\newblock Inventiones mathematicae \textbf{164}(1), 143--173 (2006).
%%\newblock \href {http://dx.doi.org/10.1007/s00222-005-0473-9}
%%  {\path{doi:10.1007/s00222-005-0473-9}}
%%
%%\bibitem{Keller}
%%Keller, B.: On differential graded categories.
%%\newblock In: Proceedings of the International Congress of Mathematicians,
%%  vol.~II, pp. 151--190. EMS, Zürich (2006)
%%
%%\bibitem{Schlichting}
%%Schlichting, M.: Higher algebraic $k$-theory.
%%\newblock In: Topics in Algebraic and Topological $K$-Theory, Lecture Notes in
%%  Math., pp. 167--241. Springer, Berlin (2011).
%%\newblock \href {http://dx.doi.org/10.1007/978-3-642-15708-0\_4}
%%  {\path{doi:10.1007/978-3-642-15708-0\_4}}
%%
%%\bibitem{Weibel}
%%Weibel, C.A.: {The $K$-book. An introduction to algebraic $K$-theory.},
%%  \emph{Graduate Studies in Mathematics}, vol. 145.
%%\newblock AMS, Providence, RI (2013)
%%
%%\end{thebibliography}

\end{document}